\documentclass[12pt]{article}
\usepackage{longtable,supertabular}
\usepackage{amsmath}
\usepackage{amssymb}
\usepackage{latexsym}
\usepackage{a4wide,color}
\usepackage{multirow}
\usepackage{ntheorem}
\usepackage[square,numbers,sort]{natbib}
\usepackage[title]{appendix}
\usepackage{array}
\usepackage{makecell}

\newtheorem{theorem}{Theorem}[section]
\newtheorem{corollary}{Corollary}[section]
\newtheorem{proposition}{Proposition}[section]
\newtheorem{lemma}{Lemma}[section]

\newtheorem*{proof}{Proof.}

\newtheorem{example}{Example}[section]

\newcommand{\bC}{{\mathbf{{C}}}}

\title{\bf Minimal Linear Codes Violating the Ashikhmin-Barg Condition from Arbitrary Projective Linear Codes}
\author{Hao Chen, Yaqi Chen, Conghui Xie and Huimin Lao\thanks{Hao Chen, Yaqi Chen and Conghui Xie are with the College of Information Science and Technology, Jinan University, Guangzhou, Guangdong Province, 510632, China (e-mail: haochen@jnu.edu.cn, chenyq@stu.jnu.edu.cn, conghui@stu2021.jnu.edu.cn). Huimin Lao is with School of Physical and Mathematical Sciences, Nanyang Technological University, Singapore (e-mail: huimin.lao@ntu.edu.sg). This research was supported by NSFC Grant 62032009.}}

\begin{document}
	
	\maketitle
	\begin{abstract}
		In recent years, there have been many constructions of minimal linear codes violating the Ashikhmin-Barg condition from Boolean functions, linear codes with few nonzero weights, posets, and partial difference sets. In this paper, we first give a general method to transform a minimal code satisfying the Ashikhmin-Barg condition to a minimal code violating the Ashikhmin-Barg condition. Then we give a construction of a minimal code satisfying the Ashikhmin-Barg condition from an arbitrary projective linear code. Hence an arbitrary projective linear code can be transformed to a minimal codes violating the Ashikhmin-Barg condition. Then we give infinitely many families of minimal codes violating the Ashikhamin-Barg condition. Weight distributions of constructed minimal codes violating the Ashikhmin-Barg condition in this paper are determined. Many minimal linear codes violating the Ashikhmin-Barg condition with their minimum weights equal to or close to the optimal or the best known minimum weights of linear codes are constructed in this paper. Moreover, many infinite families of self-orthogonal binary minimal codes violating the Ashikhmin-Barg condition are also given.\\

	\end{abstract}
	
	\newpage

\section{Introduction}

\subsection{Preliminaries}

The Hamming weight $wt({\bf a})$ of a vector ${\bf a}=(a_0, \ldots, a_{n-1}) \in {\bf F}_q^n$ is the cardinality of its support, $$supp({\bf a})=\{i: a_i \neq 0\}.$$ The Hamming distance $d({\bf a}, {\bf b})$ between two vectors ${\bf a}$ and ${\bf b}$ is $d({\bf a}, {\bf b})=wt({\bf a}-{\bf b})$. Then ${\bf F}_q^n$ is a finite Hamming metric space. The minimum (Hamming) distance of a code ${\bf C} \subset {\bf F}_q^n$ is, $$d({\bf C})=\min_{{\bf a} \neq {\bf b}} \{d({\bf a}, {\bf b}),  {\bf a} \in {\bf C}, {\bf b} \in {\bf C} \}.$$  One of the main goals of coding theory is to construct codes ${\bf C} \subset {\bf F}_q^n$ with large cardinalities and minimum distances. In general, there are some upper bounds on cardinalities or minimum distances of codes.  Optimal codes attaining these bounds are particularly interesting, see \cite{MScode}. The Singleton bound for general codes asserts
$M \leq q^{n-d+1}$ and codes attaining this bound is called (nonlinear) maximal distance separable
(MDS) codes. Reed-Solomon codes are well-known linear MDS codes, see \cite{MScode}.  For a linear code ${\bf C} \subset {\bf F}_q^n$, we denote $A_i({\bf C})$ the number of codewords with the weight $i$, $0 \leq i \leq n$. In this paper, we call a linear $[n,k,d]_q$ code optimal, if there is no linear $[n,k,d+1]_q$ code. A linear $[n,k,d]_q$ code is called almost optimal, if there is no linear $[n,k,d+2]_q$ code and there exists a linear $[n,k,d+1]_q$ code.\\

Minimal codewords were first introduced in \cite{Huang} for decoding linear block codes. Then minimal binary linear codes were studied in \cite{CL} in the name linear interesting codes. In \cite{AB}, minimal codes and minimal codewords were introduced for general $q$-ary linear codes and studied systematically.  A nonzero codeword ${\bf c}$ of a linear $[n,k,d]_q$ code ${\bf C}$ is called minimal, if for any codeword ${\bf c}'$ satisfying $supp({\bf c}') \subseteq supp({\bf c})$, then there is a nonzero $\lambda \in {\bf F}_q$, such that, ${\bf c}'=\lambda {\bf c}$. A linear code ${\bf C}$ is called minimal if each nonzero codeword is minimal.\\

The Griesmer bound for a linear $[n, k, d]_q$ code in \cite{Griesmer} asserts $$n \geq \Sigma_{i=0}^{k-1} \lceil \frac{d}{q^i} \rceil,$$
or see \cite[Theorem 2.7.4, page 81]{HP}. Set $g_q(k,d)=\Sigma_{i=0}^{k-1} \lceil \frac{d}{q^i} \rceil$, for a linear $[n,k,d]_q$ code, the Griesmer defect is $g({\bf C})=n-g_q(k,d)$. A linear $[n, k, d]_q$ code attaining is called a Griesmer code, see \cite{Solomon,Hu}. It is clear that small Griesmer defect codes are more interesting, see \cite{Solomon,Hu}.\\

\subsection{Related works}

In \cite{AB}, the Ashikhmin-Barg criterion $\frac{w_{min}}{w_{max}} > \frac{q-1}{q}$ was proposed as a sufficient condition for a $q$-ary linear code to be minimal, where $w_{min}$ and $w_{max}$ are minimum and maximum weights of the code.  However, this is certainly not a necessary condition for the minimality of linear codes. In \cite{CH} the first family of minimal binary linear codes violating the Ashikhmin-Barg criterion was constructed. Then several families of minimal binary linear codes violating the Ashikhmin-Barg criterion were constructed in \cite{DHZ} from binary linear codes with few nonzero weights. From then many minimal binary codes violating the Ashikhmin-Barg criterion or with optimal lengthes were constructed from finite geometries, Boolean functions, trace representations, posets, or partial difference sets, see \cite{DHZ,HDZ,BB,BB1,LiYue,CDMT,TangLi,Xu,Tang,Sihem2,Sihem1,Sihem,Zhang,TFL,Kan,PRZW,PRZW1,Alon,Yan} and references therein. In particular, in \cite{Hyun,Hyun1} some optimal minimal codes violating the Ashikhmin-Barg conditions were constructed. We refer to \cite{Alfarano} for combinatorial properties of minimal linear codes. \\

In the following table, we list minimal linear codes violating the Ashikhmin-Barg conditions constructed in these previous papers. Only very few of these minimal codes are close to optimal codes or best known codes in \cite{Grassl}.\\

{\footnotesize
\begin{longtable}{|c|c|>{\centering\arraybackslash}m{4.2cm}|>{\centering\arraybackslash}m{3cm}|c|}

    \caption{\label{tab:1} Known minimal $q$-ary linear codes violating the Ashikhmin-Barg condition}\\ \hline
    $[n,k]$ & Minimal distance & Condition & Method & Reference\\ \hline
    $[3^m-1,m+1]_3$ & $\sum_{j = 1}^k 2^j\binom{m}{j}$ & $2\leq k\leq\left\lfloor\frac{m - 1}{2}\right\rfloor$ & Boolean functions & \cite{HDZ} \\ \hline
    \multirow{3}{*}{$[2^m-1,m+1]_2$} & $s(2^t-1)$ & $m = 2t,s\leq2^{t - 2}$ & \multirow{3}{*}{{Boolean functions}} & \multirow{3}{*}{\cite{DHZ}} \\ \cline{2-3}
     & $2^{m-1}-2^{m-s-1}(s - 1)$ & $s=\frac{m + 1}{2}$ &  &  \\ \cline{2-3}
     & $\sum_{j = 1}^k \binom{m}{j}$ &  $2\leq k\leq\left\lfloor\frac{m - 3}{2}\right\rfloor$ & &  \\ \hline
    $[2^m-1,m+1]_2$ & $2^{m-1}+\theta_2$ & - & Partial Difference Sets & \cite{TFL} \\ \hline
    $[\binom{m}{2},m-1]_2$ & $m - 1$ & \multirow{4}{*}{{-}} & \multirow{4}{*}{-} & \multirow{4}{*}{\cite{Yan}} \\ \cline{1-2}
    $[\binom{m}{2}+1,m-1]_2$ & $m$ &  & & \\ \cline{1-2}
    $[\binom{m}{1}+\binom{m}{2},m]_2$ & $m$ & & &  \\ \cline{1-2}
    $[\binom{m}{1}+\binom{m}{2}+1,m]_2$ & $m + 1$ & & &  \\ \hline

    $[2^m-1,m+2]_2$ & $2^{\frac{n}{2}}-2$ & - & {Vectorial } & \multirow{3}{*}{\cite{Kan}} \\ \cline{1-3}
    $[2^m-1,m+n]_2$ & $2^{m-2}+2^{\frac{m}{2}+1}$ & $2 < n\leq\frac{m}{2}+1$ & Boolean & \\ \cline{1-3}
    $[2^m-1,2m+1]_2$ & $2^{m-2}+2^{\frac{m}{2}+1}$ & - & Functions &\\ \hline

    $[2^m-1,m+1]_2$ & $2^{\frac{m-1}{2}-1}(\frac{m + 1}{2}+1)-1$ & - & Characteristic & \multirow{2}{*}{\cite{Sihem}} \\ \cline{1-3}
    $[p^m-1,m+1]_p$ & - & $p$ be prime & functions &  \\ \hline

     \multirow{5}{*}{$[p^m-1,m+1]_p$} & $p^{m-s-1}(p - 1)(s(p - 1)+1)$ & \makecell{$p$ be prime,\\$(p-1)(p^{s-2}-s)>1$ } & Maiorana & \multirow{5}{*}{\cite{XQC}} \\ \cline{2-3}
     & \multirow{3}{*}{$a(p^{m-s}-p^{m-s-1})(p^k-1)$}  & \makecell{$p$ be prime,\\$k\ge2,(k,p)\neq(2,2)$,} & McFarland  &  \\
     &  & $2\le a \le(p-1)p^{k-2},s=2k$ & functions &  \\ \hline
    $[p^m-1,m]_p$ & \multirow{2}{*}{$(p - 1)^2p^{m-2}$} & \multirow{3}{*}{$m>2$, $p$ be odd prime} & \multirow{3}{*}{\parbox{3cm}{\centering$p$-ary functions}} & \multirow{3}{*}{\cite{Xu}} \\ \cline{1-1}
    $[p^m-1,m-1]_p$ &  &  & & \\ \cline{1-2}
    $[p^m-1,m]_p$ & $p^{m-1}(p - 2)$ &  & & \\ \hline

    $[q^m-1,m+1]_q$ & $\sum_{j = 1}^k (q-1)^j\binom{m}{j}$ & $q = p^h$, $p$ be odd prime, $h \ge 1$ & Boolean functions & \cite{BB1} \\ \hline

\end{longtable}
}

\subsection{Our contributions}

In this paper, we prove the following results.\\

{\bf Theorem 1.1.} {\em Let ${\bf C} \subset {\bf F}_q^n$ be a $q$-ary $k$-dimensional linear code with the maximum and minimum weights $w_{max}$ and $w_{min}$. Suppose that the Ashikhmin-Barg condition of this code ${\bf C}$ $$\frac{w_{min}}{w_{max}} >\frac{q-1}{q}$$ is satisfied and thus ${\bf C}$ is a minimal code. Set $$n'=\lceil \frac{qw_{min}}{q-1} \rceil -w_{max}.$$ Then we construct an explicit $q$-ary minimal linear code ${\bf C}' \subset {\bf F}_q^{n+n'}$ with the maximum weight at least $w_{max}+n'$ and the minimum weight $w_{min}$  satisfying $$\frac{w_{min}}{w_{max}+n'} \leq \frac{q-1}{q}.$$ The linear code ${\bf C}'$ is a minimal linear code violating the Ashikhmin-Barg condition.}\\

{\bf Theorem 1.2.} {\em Let ${\bf C} \subset {\bf F}_q^n$ be a projective linear $[n,k]_q$ code with the maximum weight $w\leq n$. Let $h$ be a positive integer. Then we construct explicit simplx complementary $[\frac{q^{k+h}-1}{q-1}-n, k+h]_q$ code with the maximum  weight $q^{k+h-1}$ and the minimum weight $q^{k+h-1}-w$. When $h >\log_q w-k+2$, this simplex complementary code is a minimal code satisfying the Ashikhmin-Barg condition.} \\

{\bf Corollary 1.1.} {\em
    Let $\mathbf{C}' \subset \mathbf{F}_q^{n+n'}$ be the minimal linear code constructed in Theorem 1.1 from an $[n, k, d]$ linear code $\mathbf{C}$ with the generator matrix ${\bf G}$ and the weight distribution $[A_0,A_{w_1},A_{w_2},\dots,A_{w_t}]$. Let $\bC^* \subset \mathbb{F}_q^n$ denote the $(k-1)$-dimensional subcode of $\bC$ generated by removing the first row of ${\bf G}$, and assume its weight distribution is $[A^*_0,A^*_{w_1}, A^*_{w_2}, \dots, A^*_{w_t}]$. Then the set of nonzero weights in $\bC'$ is $$\{w_1,w_2,\dots,w_t \} \cup \{w_1+n',w_2+n',\dots,w_t+n' \}.$$ More precisely, the weight distribution of $\bC'$ is given by:
    $$[A'_0,A'_{w_1}, \dots, A'_{w_t}, A'_{w_1 + n'}, \dots, A'_{w_t + n'}],$$
    where
    $$A'_0=1, \ A'_{w_i} = A^*_{w_i} \ \text{and} \ A'_{w_i + n'} = A_{w_i} - A^*_{w_i}, \quad 1 \le i \le t.$$\\
}

{\bf Corollary 1.2.} {\em If the code ${\bf C}\subset {\bf F}_q^n$ is a self-orthogonal linear code satisfying the Ashikhmin-Barg condition, then we can construct a minimal self-orthogonal linear code ${\bf C}' \subset {\bf F}_q^{n+n'+1}$ violating the Ashikhmin-Barg condition.}\\

The rest of this paper is organized as follows. In Section 2, we construct a minimal code violating the Ashikhmin-Barg condition from a minimal code satisfying the Ashikhmin-Barg condition. In Section 3, we construct a minimal code satisfying the Ashikhmin-Barg condition from an arbitrary projective linear code. Thus minimal linear codes violating the Ashikhmin-Barg condition can be constructed from arbitrary projective linear codes. In Section 4, we present infinitely many families of minimal linear codes violating the Ashikhmin-Barg condition. In Section 5, we construct self-orthogonal binary minimal linear codes that violate the Ashikhmin-Barg condition. In Section 6, the weight distributions of several families of  minimal codes violating the Ashikhmin-Barg condition are determined. Section 7 concludes the paper.\\

\section{Minimal codes violating the Ashikhmin-Barg condition from codes satisfying the Ashikhmin-Barg condition}

In this section, we prove the main results.\\

{\bf Theorem 2.1.} {\em Let ${\bf C} \subset {\bf F}_q^n$ be a $q$-ary $k$-dimensional linear code with the maximum and minimum weights $w_{max}$ and $w_{min}$. Suppose that the Ashikhmin-Barg condition of this code ${\bf C}$ $$\frac{w_{min}}{w_{max}} >\frac{q-1}{q}$$ is satisfied and thus ${\bf C}$ is a minimal code. Set $$n'=\lceil \frac{qw_{min}}{q-1} \rceil -w_{max}.$$ Then we construct an explicit $q$-ary minimal linear code ${\bf C}' \subset {\bf F}_q^{n+n'}$ with the maximum weight at least $w_{max}+n'$ and the minimum weight $w_{min}$  satisfying $$\frac{w_{min}}{w_{max}+n'} \leq \frac{q-1}{q}.$$}\\

{\bf Proof.} Let ${\bf G}$ be a generator matrix of the linear code ${\bf C}$, with the first row ${\bf r}_1$ equal to a codeword $(c_1,\ldots, c_n)$ with the maximum weight $w_{max}$, and the second row ${\bf r}_2$ equal to a codeword with the minimum weight $w_{min}$. Let ${\bf G}'$ be a $k \times (n+n')$ matrix with the first row ${\bf r}_1'$ equal to $(c_1',\ldots, c_n',c_1, \ldots, c_n)$, where $c_1', \ldots, c_n'$ are nonzero elements in ${\bf F}_q$, and the other rows are of the form ${\bf r}_i'=({\bf 0}, {\bf r}_i)$, where ${\bf 0}$ is the zero vector in ${\bf F}_q^{n'}$, ${\bf r}_i$ are $i$-th row of ${\bf G}$, $i=2,\ldots,k$.\\

Now it is clear that the matrix ${\bf G}'$ is full rank and the linear code in ${\bf F}_q^{n+n'}$ generated by ${\bf G}'$ is the linear code ${\bf C}'$. The maximum weight of ${\bf C}'$ is at least $n'+w_{max}$, since the Hamming weight of the first row of ${\bf G}'$ is $w_{max}+n'$.  We now prove that ${\bf C}'$ is a minimal code. Let ${\bf x}_1$ and ${\bf x}_2$ be two nonzero codewords in ${\bf C}'$ satisfying $$supp({\bf x}_2) \subseteq supp({\bf x}_1),$$ we analysis the following two cases.\\

1) ${\bf x}_1$ is the linear combination of rows ${\bf r}_2',\ldots, {\bf r}_k'$ of ${\bf G}'$. In this case, the support of ${\bf x}_1$ is in the last $n$ coordinate positions. Then the support of ${\bf x}_2$ is also in the last $n$ coordinate positions. The codeword ${\bf x}_2$ is a linear combination of ${\bf r}_2',\ldots, {\bf r}_k'$ rows of ${\bf G}'$. Then from the minimality of the linear code ${\bf C}$, there exists a nonzero $\lambda \in {\bf F}_q$ such that ${\bf x}_2=\lambda{\bf x}_1$.\\

2) ${\bf x}_1$ is the linear combination of rows ${\bf r}_1', \ldots, {\bf r}_k'$ of ${\bf G}'$. Suppose that $${\bf x}_1=\mu_1{\bf r}_1'+\mu_2 {\bf c},$$ where $\mu_1\neq 0$, and ${\bf c}$ is a linear combination of rows ${\bf r}_2',\ldots, {\bf r}_k'$ in ${\bf G}'$. Firstly, the support of ${\bf x}_2$ cannot be in the last $n$ coordinate positions. Let ${\bf x}_i'$ be the vector in ${\bf F}_q^n$ by puncturing the first $n'$ coordinates of ${\bf x}_i$, $i=1,2$. They are two nonzero codewords of ${\bf C}$, and we have $supp({\bf x}_2') \subset  supp({\bf x}_1')$. Then ${\bf x}_2'$ is of the form $\lambda'{\bf x}_1'$, where $\lambda'\neq 0$, from the minimality of the code ${\bf C}$. Therefore, ${\bf x}_2$ cannot be a linear combination of only ${\bf r}_2', \ldots, {\bf r}_k'$ in the generator matrix ${\bf G}'$. Then the support of ${\bf x}_2$ has to include the first $n'$ coordinate positions. From a similar argument as above, we also conclude that there exists a nonzero $\lambda''$ satisfying $${\bf x}_2=\lambda''{\bf x}_1.$$ Then the linear code ${\bf C}'$ is a minimal.\\

The maximum weight of the code ${\bf C}'$ is at least $$w_{max}({\bf C}') \geq w_{max}+n'.$$ It is easy to verify that $$\frac{w_{\min}({\bf C}')}{w_{max}({\bf C}')} \leq \frac{w_{min}}{w_{max}+n'} \leq \frac{q-1}{q}.$$ The minimal code ${\bf C}'$ violates the Ashikhmin-Barg condition. The conclusion is proved.\\

{\bf Corollary 2.1.} {\em If the code ${\bf C}\subset {\bf F}_q^n$ is a self-orthogonal linear code satisfying the Ashikhmin-Barg condition, then we can construct a minimal self-orthogonal linear code ${\bf C}' \subset {\bf F}_q^{n+n'+1}$ violating the Ashikhmin-Barg condition.}\\

{\bf Proof.} Since the code ${\bf C}$ is self-orthogonal, we only need to impose a condition on coordinates of $n'$ positions to guarantee the linear code ${\bf C}'$ is self-orthogonal.\\





\section{Minimal codes satisfying the Ashikhmin-Barg condition from arbitrary projective linear codes}

In this section, we prove the main result Theorem 3.1. The construction is based on our previous construction in \cite{ChenXie}. We recall the simplex complementary code of a projective linear code introduced in \cite{ChenXie}.\\

The simplex complementary code of a projective linear $[n,k]_q$ code ${\bf C}$ with the maximum weight $w$ is defined as follows. Here the length satisfies the condition $n<q^{k-1}$.  Let ${\bf G}$ be the $k \times n$ generator matrix of the projective linear code ${\bf C}$. Then $n$ columns ${\bf g}_1, \ldots, {\bf g}_n$ of ${\bf G}$ are distinct vectors of ${\bf F}_q^k$. By deleting these $n$ distinct columns in the generator matrix of the $q$-ary simplex $[\frac{q^k-1}{q-1},k, q^{k-1}]_q$ code, we get a linear $[\frac{q^k-1}{q-1}-n, k, q^{k-1}-w]_q$ code. The dimension is $k$, because there are $\frac{q^k-1}{q-1}-n>\frac{q^{k-1}-1}{q-1}$ columns in the generator matrix of this code. These columns cannot be located in a $k-1$ dimension linear subspace.\\

More generally, let ${\bf g}_1, \ldots, {\bf g}_n$ be $n$ distinct columns of the projective linear $[n,k]_q$ code ${\bf C} \subset {\bf F}_q^n$. By appending some zero coordinates, we can consider ${\bf g}_1, \ldots, {\bf g}_n$ as $n$ distinct vectors in ${\bf F}_q^K$, where $K>k$. By deleting these columns in $\frac{q^K-1}{q-1}$ columns of the $q$-ary simplex $[\frac{q^K-1}{q-1},K, q^{K-1}]_q$ code, we get a projective $[\frac{q^K-1}{q-1}-n, K, q^{K-1}-w]_q$ code. It is easy to verify that the maximum weight of this simplex complementary $[\frac{q^K-1}{q-1}-n, K, q^{K-1}-w]_q$ code is $q^{K-1}$.\\

\begin{theorem}
    Let ${\bf C} \subset {\bf F}_q^n$ be a projective linear $[n,k]_q$ code with the maximum weight $w\leq n$. Let $h$ be a positive integer. Then the simplx complementary code of ${\bf C}$ has parameters $[\frac{q^{k+h}-1}{q-1}-n, k+h, q^{k+h-1}-w]_q$ and maximum weight $q^{k+h-1}$. When $h >\log_q w-k+2$, this simplex complementary code is a minimal code satisfying the Ashikhmin-Barg condition.
\end{theorem}

\begin{proof}
    {\rm It is easy to verify that if $h \geq log_qw-k+2$, we have $$\frac{q^{k+h-1}-w}{q^{k+h-1}}>\frac{q-1}{q}.$$ Then this simplex complementary $[\frac{q^{k+h}-1}{q-1}-n, k+h, q^{k+h-1}-w]_q$ code is a minimal code satisfying the Ashikhmin-Barg condition. The conclusion is proved.
    }
\end{proof}

By Theorem 2.1 and Theorem 3.1, we have the following corollary.

\begin{corollary}
    From the simple complementary $[\frac{q^{k+h}-1}{q-1}-n, k+h, q^{k+h-1}-w]_q$ code with the maximum weight $q^{k+h-1}$, we construct a minimal code with parameters $[\frac{q^{k+h}-1}{q-1}-n+n', k+h, q^{k+h-1}-w]_q$ and maximum weight $q^{k+h-1}+n'$. It is a minimal code that violates the Ashikhmin-Barg condition.
\end{corollary}

\section{Infinitely many families of minimal linear codes violating the Ashikhmin-Barg codndition}

In this section, we present infinitely many families of explicit minimal linear codes that violate the Ashikhmin-Barg condition. Firstly, we directly transform known minimal codes satisfying the Ashikhmin-Barg condition into minimal codes violating the Ashikhmin-Barg condition via Theorem 2.1. Secondly, we employ Theorem 3.1 to convert arbitrary projective linear codes into minimal codes satisfying the Ashikhamin-Barg condition,  which are then transformed to minimal codes violating the Ashikhmin-Barg condition from Theorem 2.1.

\subsection{From minimal Solomon-Stiffler codes}

In this subsection, we construct a family of binary minimal linear codes violating the Ashikhmin-Barg condition from binary minimal Solomon-Stiffler codes.\\

The binary Solomon-Stiffler code is constructed by removing carefully selected columns from the generator matrix of the binary simplex $[2^k - 1,\ k,\ 2^{k-1}]_2$ code. Let $u_1, u_2, \dots, u_t$ be integers such that $k > u_t > \cdots > u_1 \geq 1$, and let $S_1, \dots, S_t$ be mutually disjoint subspaces of $\mathbb{F}_2^k$ with $\dim(S_i) = u_i$. Collecting all nonzero vectors from each $S_i$ forms a set $\mathcal{F}$ of pairwise disjoint columns. By removing the columns indexed by $\mathcal{F}$ from the simplex generator matrix, we obtain a binary linear code with parameters
$$
[2^k - 1 - \sum_{i=1}^t (2^{u_i} - 1),\ k,\ 2^{k-1} - \sum_{i=1}^t 2^{u_i - 1}]_2,
$$
which meets the Griesmer bound.\\

\begin{proposition}\label{P-4.1}
    Let $u_1, u_2, \dots, u_t$ be integers satisfying $k-1 > u_t > \cdots > u_1 \geq 1$, $\sum_{i=1}^{t}u_i <k$. The binary Solomon-Stiffler $[2^k - 1 - \sum_{i=1}^t (2^{u_i} - 1),\ k,\ 2^{k-1} - \sum_{i=1}^t 2^{u_i - 1}]_2$ code is a minimal code satisfying the Ashikhmin-Barg condition. A family of binary minimal codes with parameters $[2^k +2^{k-1} + t - 1 - \sum_{i=1}^t 2^{u_i+1},\ k,\ 2^{k-1} - \sum_{i=1}^t 2^{u_i - 1}]_2$ is constructed from these binary Solomon-Stiffler codes. These codes are minimal codes violating the Ashikhmin-Barg condition.
\end{proposition}

{\bf Proof.} When $\sum_{i=1}^{t}u_i <k$, $u_t<k-1$, the maximum weight is $2^{k-1}$. Therefore, the binary Solomon-Stiffler codes are minimal codes satisfying the Ashikhmin-Barg condition, since $$\frac{2^{k-1}-\sum_{i=1}^{t}2^{u_i-1} }{2^{k-1}} >\frac{1}{2}.$$ Then, minimal codes violating the Ashikhmin-Barg condition can be constructed from binary Solomon-Stiffler codes, based on the construction of Theorem 2.1.\\

\begin{example}{\rm
    When $k=4$, $t=2$, $(u_1,u_2)=(1,2)$, the corresponding Solomon-Stiffler $[11,4,5]_2$ code has the maximum weight $8$. Then a minimal $[13,4,5]_2$ code with the maximum weight $10$ is constructed from Theorem 2.1. This minimal code violates the Ashikhmin-Barg condition and is almost optimal, see \cite{Grassl}.}
\end{example}

\begin{example}{\rm When $k=5$, $t=2$, $(u_1,u_2)=(1,3)$, the corresponding Solomon-Stiffler code is a binary linear $[23,5,11]_2$ code with the maximum weight $16$. Then a minimal $[29,5,11]_2$ code with the maximum weight $22$ is constructed from this $[23,5,11]_2$ code, from Theorem 2.1. This minimal code violates the Ashikhmin-Barg condition. The optimal weight of a linear $[29,5]_2$ code is $14$, see \cite{Grassl}.}
\end{example}

Table \ref{tab:Solomon} lists the binary minimal codes constructed in Proposition \ref{P-4.1}. Here, $\bC$ denotes the original Solomon-Stiffler code satisfying the Ashikhmin-Barg condition,  while $\bC'$ denotes the minimal code that violates the Ashikhmin-Barg condition from Theorem 2.1. Notably, in cases with small dimension, the minimum weight of $\bC'$ is close to the best known weight of the corresponding linear code in \cite{Grassl}.

\begin{longtable}{|l|l|l|}
    \caption{\label{tab:Solomon} Binary Minimal codes violating the AB condition constructed in Proposition \ref{P-4.1} }\\ \hline
    $\bC$ & $\bC'$ & Optimal or best known Parameters \\ \hline
    $[6,3,3]_2$  & $[8,3,3]_2$  & $[8,3,4]_2$ \\ \hline
    $[11,4,5]_2$  & $[13,4,5]_2$  & $[13,4,6]_2$ \\ \hline
    $[12,4,6]_2$  & $[16,4,6]_2$  & $[16,4,8]_2$ \\ \hline
    $[14,4,7]_2$  & $[20,4,7]_2$  & $[20,4,10]_2$ \\ \hline
    $[23,5,11]_2$  & $[29,5,11]_2$  & $[29,5,14]_2$ \\ \hline
    $[24,5,12]_2$  & $[32,5,12]_2$  & $[32,5,16]_2$ \\ \hline
    $[27,5,13]_2$  & $[37,5,13]_2$  & $[37,5,18]_2$ \\ \hline
    $[47,6,23]_2$  & $[61,6,23]_2$  & $[61,6,30]_2$ \\ \hline
    $[48,6,24]_2$  & $[64,6,24]_2$  & $[64,6,32]_2$ \\ \hline
    $[55,6,27]_2$  & $[77,6,27]_2$  & $[77,6,38]_2$ \\ \hline
    $[95,7,47]_2$  & $[125,7,47]_2$  & $[125,7,62]_2$ \\ \hline
    $[96,7,48]_2$  & $[128,7,48]_2$  & $[128,7,64]_2$ \\ \hline
    $[107,7,54]_2$  & $[153,7,54]_2$  & $[153,7,76]_2$ \\ \hline
    $[292,8,96]_2$  & $[256,8,96]_2$  & $[256,8,128]_2$ \\ \hline
\end{longtable}	

\subsection{From few-weight codes satisfying the Ashikhmin-Barg condition}
In this subsection, we construct several infinite families of minimal $q$-ary linear codes that violate the Ashikhmin–Barg condition, derived from known few-weight codes that satisfy the  Ashikhmin–Barg condition.\\

Let $m \ge 2$ be an integer. A family of three-weight projective linear $[2^{2m}-1, 3m, 2^{2m-1}-2^{m-1}]_2$ codes is constructed in \cite{Kasami}. The three nonzero weights are $2^{2m-1}-2^{m-1}, 2^{2m-1}$ and $2^{2m-1}+2^{m-1}$. These codes are minimal codes satisfying the Ashikhmin-Barg condition.  From Theorem 2.1, we have the following result.\\

\begin{proposition}\label{P-4.2}
    Let $m \geq 2$ be an integer and $n' = 2^{2m-1} - 2^m - 2^{m-1}$. A family of minimal linear codes with parameters $[2^{2m}-1+n', 3m, 2^{2m-1}-2^{m-1}]_2$ and the maximum weight $2^{2m}-2^{m}$ is constructed from three-weight minimal linear codes with parameters $[2^{2m}-1, 3m, 2^{2m-1}-2^{m-1}]_2$. These codes are minimal codes violating the Ashikhmin-Barg condition.
\end{proposition}

{\bf Proof.} The maximum weight of the code is known, then the conclusion follows from the construction in Theorem 2.1.\\

\begin{example}{\rm When $m=2$ and $n'=2$, a minimal linear $[17, 6, 6]_2$ code with the maximum weight $12$ is constructed from the three-weight minimal linear $[15, 6, 6]_2$ code with the maximum weight $10$. This code is a minimal code violating the Ashikhmin-Barg condition. Moreover, this code is almost optimal, see \cite{Grassl}.}
\end{example}

The simplex complementary codes of $[2^{2m}-1, 3m, 2^{2m-1}-2^{m-1}]_2$ codes are $[2^{3m+h}-2^{2m},3m+h,2^{3m+h-1}-2^{2m-1}-2^{m-1}]_2$ codes with the maximum weight $2^{3m+h-1}$, see \cite[Theorem 5.2]{ChenXiePan}.  Similarly, from Theorem 2.1, we have the following result.\\

\begin{proposition}\label{P-4.3}
    Let $m \geq 2$ and $h$ be positive integers, $n' = 2^{3m+h-1} - 2^{2m} - 2^{m}$. A family of minimal linear codes with parameters $[2^{3m+h}-2^{2m}+n',3m+h,2^{3m+h-1}-2^{2m-1}-2^{m-1}]_2$ and the maximum weight $2^{3m+h}-2^{2m}-2^{m}$ is constructed from four-weight minimal linear codes with parameters $[2^{3m+h}-2^{2m},3m+h,2^{3m+h-1}-2^{2m-1}-2^{m-1}]_2$. These codes are minimal codes violating the Ashikhmin-Barg condition.
\end{proposition}

\begin{example}
    {\rm When $m=2$, $h=1$ and $n'=44$, a minimal linear $[156, 7, 54]_2$ code with the maximum weight $108$ is constructed from the three-weight minimal linear $[112, 7, 54]_2$ code with the maximum weight $64$. This code is a minimal code violating the Ashikhmin-Barg condition.

    }
\end{example}

Let $\frac{m}{l}$ be odd. A family of three-weight projective linear $[2^{m-1}-1,m, 2^{m-2}-2^{\frac{m+l-4}{2}}]_2$ codes was constructed in \cite[Theorem 1]{DD1}. The three nonzero weights are $2^{m-2}-2^{\frac{m+l-4}{2}}, 2^{m-2}$ and $2^{m-2}+2^{\frac{m+l-4}{2}}$. From Theorem 2.1, we have the following result.\\

\begin{proposition}\label{P-4.4}
    Let $m$ and $l$ be positive integers such that $\frac{m}{l}$ is odd. Set $n' = 2^{m-2}-3\cdot 2^{\frac{m+l-4}{2}}$. A family of $[2^{m-1}-1+n',m, 2^{m-2}-2^{\frac{m+l-4}{2}}]_2$ codes with the maximum weight $2^{m-1}-2^{\frac{m+l-2}{2}}$ is constructed. These codes are minimal codes violating the Ashikhmin-Barg condition.
\end{proposition}

\begin{example}
    {\rm When $(m, l)=(5, 1)$, a minimal linear $[17,5,6]_2$ code with the maximum weight $12$ is constructed from the three-weight minimal linear $[15,5,6]_2$ code with the maximum weight $10$. The corresponding optimal linear code in \cite{Grassl} is a linear $[17,5,8]_2$ code.}
\end{example}

\begin{example}
    {\rm When $(m, l)=(6, 2)$, a minimal linear $[35,6,12]_2$ code with the maximum weight $24$ is constructed from the three-weight minimal linear $[31,6,12]_2$ code with the maximum weight $20$. The corresponding optimal linear code in \cite{Grassl} is a linear $[35,6,16]_2$ code.}
\end{example}

Let $m\geq5$ be an odd positive integer and $h$ be a positive integer. The simplex complementary codes of $[2^{m-2}, m-1, 2^{m-3}-2^{\frac{m-3}{2}}]_2$ codes are $[2^{m+h-1}-2^{m-2}-1, m+h-1, 2^{m+h-2}-2^{m-3}-2^{\frac{m-3}{2}}]_2$ codes with maximum weight $2^{m+h-2}$, see \cite[Theorem 5.7]{ChenXiePan}. From Theorem 2.1, we have the following result.\\

\begin{proposition}\label{P-4.5}
    Let $m\geq5$ be an odd positive integer and $h$ be a positive integer. Set $n'=2^{m+h-2} - 2^{m-2} - 2^{\frac{m-1}{2}}$. A family of $[2^{m+h-1}-2^{m-2}-1+n', m+h-1, 2^{m+h-2}-2^{m-3}-2^{\frac{m-3}{2}}]_2$ codes with the maximum weight $2^{m+h-2}$ is constructed. These codes are minimal codes violating the Ashikhmin-Barg condition.\\
\end{proposition}

\begin{example}
    {\rm When $(m,h)=(5,1)$, a minimal linear $[27,5,10]_2$ code with the maximum weight $20$ is constructed from the three-weight minimal linear $[23,5,10]_2$ code with maximum weight $16$. The optimal linear $[27,5,13]_2$ code was documented in \cite{Grassl}.}
\end{example}

Let $h$ be a positive integer. The simplex complementary code of the linear $[2^{2m-3}+2^{m-2}-1,2m-2,2^{2m-4}]_2$ code is a linear $[2^{2m+h-2}-2^{2m-3}-2^{m-2},2m+h-2,2^{2m+h-3}-2^{2m-4}-2^{m-2}]_2$ code with the maximum weight $2^{2m+h-3}$, see \cite[Theorem 5.8]{ChenXiePan}. From Theorem 2.1, we have the following result.\\

\begin{proposition}\label{P-4.7}
    Let $m\geq3$ be a positive integer, $h$ be a positive integer. Set $n'=2^{2m+t-3} - 2^{2m-3} - 2^{m-1}$. A family of $[2^{2m+h-2}-2^{2m-3}-2^{m-2}+n',2m+h-2,2^{2m+h-3}-2^{2m-4}-2^{m-2}]_2$ codes with the maximum weight $2^{2m+h-2}-2^{2m-3}-2^{m-1}$ is constructed. These codes are minimal codes violating the Ashikhmin-Barg condition.
\end{proposition}

\begin{example}
    {\rm When $m=3$, a minimal linear $[26,5,10]_2$ code with the maximum weight $20$ is constructed from the two-weight minimal linear $[22,5,10]_2$ code with the maximum weight $16$. The optimal linear $[26,5,12]_2$ code was documented in \cite{Grassl}.}
\end{example}

Let $m$ be an odd positive integer. Two families of projective linear $[n_1, m, \frac{n_1-2^{(m-1)/2}}{2}]_2$ code with three nonzero weights, $\frac{n_1-2^{(m-1)/2}}{2}$, $\frac{n_1}{2}$ and $\frac{n_1+2^{(m-1)/2}}{2}$, were constructed in \cite[Corollary 11]{Ding}, where $n_1=2^{m-1}$ or $n_1=2^{m-1}+2^{(m-1)/2}$. For $m=5$ and $n_1=2^{m-1}$, the linear $[16,5,6]_2$ code has the maximum weight $10$. From Theorem 3.1, a minimal linear $[18,5,6]_2$ code with the maximum weight $12$ is obtained. For $m=5$ and $n_1=2^{m-1}+2^{(m-1)/2}$, the binary $[20,5,8]_2$ code has the maximum weight $12$. From Theorem 3.1, a minimal linear $[24,5,8]_2$ code with the maximum weight $16$ is obtained. The optimal weights of a linear $[18, 5]_2$ code is $8$ and the optimal minimum weight of  a linear $[24,5]_2$ code is $12$.\\

\begin{proposition}\label{P-4.8}
    Let $m$ be an odd positive integer, $n_1=2^{m-1}$ or $n_1=2^{m-1}+2^{(m-1)/2}$. A family of $[\frac{3}{2}\cdot(n_1-2^{\frac{m-1}{2}}), m, \frac{n_1-2^{(m-1)/2}}{2}]_2$ codes with the maximum weight $n_1-2^{\frac{m-1}{2}}$ is constructed. These codes are minimal codes violating the Ashikhmin-Barg condition.
\end{proposition}

The simplex complementary code of the linear $[n_1, m, \frac{n_1-2^{(m-1)/2}}{2}]_2$ code is a linear $[2^{m+t}-n_1-1, m+t, 2^{m+t-1}-\frac{n_1+2^{(m-1)/2}}{2}]_2$ code. For example, from projective liner $[11,4,5]_2$, $[25,5,12]_2$ and $[27,5,13]_2$ codes, minimal linear $[13,4,5]_2$, $[32,5,12]_2$ and $[37,5,13]_2$ codes violating the Ashikhmin-Barg condition are constructed. The corresponding optimal weights are $6$, $16$ and $18$, see \cite{Grassl}.\\

Similarly, we have the following result.\\

\begin{proposition}\label{P-4.9}
    Let $m$ be an odd positive integer, $n_1=2^{m-1}$ or $n_1=2^{m-1}+2^{(m-1)/2}$. A family of $[3\cdot2^{m+t-1}-2n_1-2^{\frac{m-1}{2}}-1, m+t, 2^{m+t-1}-\frac{n_1+2^{(m-1)/2}}{2}]_2$ codes with the maximum weight $2^{m+t}-n_1-2^{\frac{m-1}{2}}$ is constructed. These codes are minimal codes violating the Ashikhmin-Barg condition.
\end{proposition}

Let $m\geq 3$ be an odd integer and $q$ be an odd prime. A family of three-weight codes with parameters $[q^{m-1}, m, q^{m-1}-q^{m-2}-q^{\frac{m-3}{2}}]_q$ was constructed in \cite[Corollary 3]{DD}. Their simplex complementary codes are $[q^{m-1}, m, q^{m-1}-q^{m-2}-q^{\frac{m-3}{2}}]_q$ codes with weights $q^{m-1}-q^{m-2}+q^{\frac{m-3}{2}}$, $q^{m-1}-q^{m-2}$ and $q^{m-1}-q^{m-2}-q^{\frac{m-3}{2}}$. When $q=3,5$ and $m=3$, a minimal linear $[9,3,5]_3$ code with the maximum weight $7$ and a minimal linear $[23,5,19]_5$ code with the maximum weight $21$ are obtained. From Corollary 3.1, an almost optimal minimal $[10,3,5]_3$ code violating the Ashikhmin-Barg condition and the minimal linear $[28,3,19]_5$ code violating the Ashikhmin-Barg condition are constructed. The optimal minimum weight of a linear $[28,3]_5$ code is $22$, see \cite{Grassl}.

\begin{proposition}\label{P-4.10}
    Let $m\geq 3$ be an odd integer and $q$ be an odd prime. Set $$n'=\left \lceil \frac{q}{q-1} \cdot (q^{m-1}-q^{m-2}-q^{\frac{m-3}{2}}) \right \rceil - q^{m-1}+q^{m-2}-q^{\frac{m-3}{2}}. $$ A family of $[q^{m-1}+n', m, q^{m-1}-q^{m-2}-q^{\frac{m-3}{2}}]_q$ codes with the maximum weight $q^{m-1}-q^{m-2}+q^{\frac{m-3}{2}}+n'$ is constructed. These codes are minimal codes violating the Ashikhmin-Barg condition.
\end{proposition}

\subsection{From linear codes violating the Ashikhmin-Barg condition}
In this subsection, we construct minimal $q$-ary linear codes that violate the Ashikhmin–Barg condition, based on the construction given in Theorem 3.1.

Let $\bC_0$ be a linear $[n, n-1, 2]_2$ code with maximum weight $n$ if $n$ is even, and $n-1$ if $n$ is odd. This code violates the Ashikhmin-Barg condition, is not a minimal code. From Theorem 3.1 and 2.1, we can transform it into a minimal code violating the Ashikhmin-Barg condition.

\begin{proposition}\label{P-4.11}  Let $h$ be a non-negative integer.
When
$$
h > \begin{cases}
\log_2(n - 1) - n + 3, & \text{if } n \text{ is odd}, \\
\log_2 n - n + 3, & \text{if } n \text{ is even},
\end{cases}
$$
the simplex complementary code $\bC_1$ of $\bC_0$ has parameters
$$
\begin{cases}
\left[2^{n + h - 1} - n - 1,\; n + h - 1,\; 2^{n + h - 2} - n + 1\right]_2, & \text{if } n \text{ is odd}, \\
\left[2^{n + h - 1} - n - 1,\; n + h - 1,\; 2^{n + h - 2} - n\right]_2, & \text{if } n \text{ is even},
\end{cases}
$$
and the maximum weight $2^{n + h - 2}$. From Theorem 2.1, the minimal linear codes $\bC'$ violating the Ashikhmin-Barg condition is constructed with parameters
$$
\begin{cases}
\left[3\cdot2^{n + h - 2} - 3n + 1,\; n + h - 1,\; 2^{n + h - 2} - n + 1\right]_2, & \text{if } n \text{ is odd}, \\
\left[3\cdot2^{n + h - 2} - 3n - 1,\; n + h - 1,\; 2^{n + h - 2} - n\right]_2, & \text{if } n \text{ is even},
\end{cases}
$$
The maximum codeword weight is $2^{n + h - 1} - 2n + 2$ if $n$ is odd and $2^{n + h - 1} - 2n $ if $n$ is even.\\
\end{proposition}

{\bf Proof.} These parameters can be calculated from constructions in Theorem 3.1 and 2.1.

\begin{example}
    {\rm  When $n=4$, let $\bC_0$ be the linear $[4,3,2]_2$ code. The simplex complementary code $\bC_1$ of $\bC_0$ is the $[27,5,12]_2$ code with the maximum weight $16$. From Theorem 2.1, $\bC_1$ can be transformed into the $[35,5,12]_2$ code with the maximum weight $24$. It is a minimal code that violates the Ashikhmin-Barg condition. The corresponding optimal linear codes in \cite{Grassl} is a linear $[35,5,16]_2$ code. Similarly, a minimal linear $[34,5,12]_2$ code and a minimal linear $[31,5,10]_2$ code violating the Ashikhmin-Barg condition are obtained from the linear $[5,4,2]_2$ code and the linear $[6,5,2]_2$ code.}
\end{example}

Table \ref{tab:4-nAB} lists these minimal codes constructed from Proposition \ref{P-4.11}, where $\bC_0$ denotes the original code, $\bC_1$ denotes the simplex complementary code of $\bC_0$ and $\bC'$ denotes the minimal codes violating the Ashikhmin-Barg condition constructed from $\bC_1$.\\

\begin{longtable}{|l|l|l|l|l|}

\caption{\label{tab:4-nAB} Minimal codes violating the AB condition constructed in Proposition \ref{P-4.11}}\\ \hline
$\bC_0$ & $\bC_1$ & $\bC'$ & Optimal parameters   \\ \hline
$[4,3,2]_2$ &  $[27,5,12]_2$ & $[35,5,12]$ & $[35,5,16]_2$ \\ \hline
$[5,4,2]_2$ &  $[26,5,12]_2$ & $[34,5,12]$ & $[34,5,16]_2$ \\ \hline
$[6,5,2]_2$ &  $[25,5,10]_2$ & $[31,5,10]$ & $[31,5,16]_2$ \\ \hline

\end{longtable}

\subsection{From arbitrary projective linear codes}

In this subsection, we construct minimal linear codes that violate the Ashikhmin-Barg condition from arbitrary projective linear codes.\\

Given any best known $q$-ary projective linear code $\bC_0$ from \cite{Grassl}, we can calculate its weight distribution with Magma. If $\bC_0$ satisfies the Ashikhmin-Barg condition, Theorem 2.1 can be directly applied to construct a minimal code $\bC'$ violating the Ashikhmin-Barg condition. Otherwise, Theorem 3.1 is used to construct a minimal code $\bC_1$ satisfying the Ashikhmin-Barg condition, from the simplex complement code of the code $\bC_0$. Then Theorem 2.1 is applied to obtain a minimal code $\bC'$ violating the Ashikhmin-Barg condition. In particular, when the original code is optimal and exhibits a significant gap between its minimum and maximum weights, the resulting minimal code not only violates the Ashikhamin-Barg condition, but also has its minimum weight close to the best known weight of the corresponding linear code in \cite{Grassl}.

\begin{example}
{\rm
(1) The binary $[43,9,18]_2$ code $\bC_0$ with the maximum weight $32$ in \cite{Grassl} satisfies the Ashikhmin-Barg condition. The resulting code $\bC'$ is a minimal linear $[47,9,18]_2$ code with the maximum weight $36$. The minimum weight $18$ is close to the best known weight $20$ of the linear $[47,9,20]_2$ code in \cite{Grassl}.\\

(2)The ternary $[182,15,99]_3$ code $\bC_0$ with the maximum weight $147$ in \cite{Grassl} satisfies the Ashikhmin-Barg condition. The resulting $[184,15,99]_3$ code $\bC'$  has the maximum weight $149$. This code is optimal, see \cite{Grassl}.

(3) The ternary $[208,14,117]_3$ code $\bC_0$ with the maximum weight $168$ in \cite{Grassl} satisfies the Ashikhmin-Barg condition. The resulting $[216,14,117]_3$ code $\bC'$ has the maximum weight $176$. The minimum weight $117$ is close to the minimum weight $120$ of the best known linear $[216,14,120]_3$ code, see \cite{Grassl}.

(4) The quasi-cyclic $[62,6,45]_5$ code $\bC_0$ with maximum weight $55$ in \cite{Grassl} satisfies the Ashikhmin-Barg condition. The resulting $[64,6,45]_5$ code $\bC'$ has the maximum weight $57$. This code is best known, see \cite{Grassl}.

(5) The ternary $[7,3,4]_3$ code $\bC_0$ in \cite{Grassl} violates the Ashikhmin-Barg condition. The simplex complementary code $\bC_1$ of $\bC_0$ is a linear $[33,4,20]_3$ code and has the maximum weight $27$. This code satisfies the Ashikhmin-Barg condition. The resulting $[36,4,20]_3$ code $\bC'$ has the maximum weight $30$. The minimum weight $20$  is close to the minimum weight of the optimal $[36,4,24]_3$ code in \cite{Grassl}.
}

\end{example}

Further examples obtained from the best known codes in the codetable \cite{Grassl} are presented in Table \ref{tab:from-codetable} in Appendix A.\\

\section{Infinitely many families of self-orthogonal binary minimal linear codes violating the Ashikhmin-Barg condition}

A binary linear code is said to be doubly even if the Hamming weight of every nonzero codeword is divisible by four. Such codes are necessarily self-orthogonal with respect to the standard inner product. Notice that the simplex complementary code of a doubly even binary code is doubly even again. In this section, we present infinitely many infinite families of doubly even minimal linear codes that violate the Ashikhmin–Barg condition.\\

Starting from a known doubly even binary cyclic code $\mathbf{C}$, we apply Corollary 3.1 to construct a doubly even binary minimal code $\mathbf{C}'$ that violates the Ashikhmin–Barg condition. Both $\mathbf{C}$ and $\mathbf{C}'$ are listed in Table \ref{tab:5-cyclic}.

\begin{longtable}{|l|l|l|l|}

\caption{\label{tab:5-cyclic} Binary self-orthogonal minimal codes violating the AB condition constructed from binary self-orthogonal cyclic codes}\\ \hline
$\bf C$       & $\bf C'$ & Weights of $\bf C'$ & Best known parameters   \\ \hline
$[31,10,12]_2$ &  $[35,10,12]_2$ & $16,20,24$ & $[35,10,16]_2$ \\ \hline
$[63,12,24]_2$ &  $[71,12,24]_2$ & $32,40,48$ & $[71,12,29]_2$ \\ \hline
$[63,15,24]_2$ &  $[71,15,24]_2$ & $32,36,40,44,48$ & $[71,15,26]_2$ \\ \hline
$[85,16,32]_2$ &  $[93,16,32]_2$ & $40,44,48,52,56,60,64$ & $[93,16,35]_2$ \\ \hline
$[93,15,32]_2$ & $[101,15,32]_2$ & $40,48,56,64$ & $[101,15,40]_2$ \\ \hline
$[127,21,48]_2$ & $[143,21,48]_2$ & $56,64,72,80,88$ & $[143,21,54]_2$ \\ \hline
$[127,28,44]_2$ & $[131,28,44]_2$ & $48,52,56,60,64,68,72,76,80,84,88$ & $[131,28,44]_2$ \\ \hline
$[195,24,72]_2$ & $[211,24,72]_2$ & $88,96,104,112,120,128,136,144$ & $[211,24,80]_2$ \\ \hline

\end{longtable}

\begin{lemma}\citep[Theorem 4.3]{SH}\label{L-5.1}
    Let $k \ge 4$ and $u_1, u_2, \dots, u_t$ be integers satisfying $k > u_t > \cdots > u_1 \geq 3$, the binary Solomon-Stiffler $[2^k - 1 - \sum_{i=1}^t (2^{u_i} - 1),\ k,\ 2^{k-1} - \sum_{i=1}^t 2^{u_i - 1}]_2$ code is self-orthogonal.
\end{lemma}

When $k \ge 5$ and $k-1 > u_t > \cdots > u_1 \geq 3$, the binary Solomon-Stiffler code is minimal code satisfying the Ashikhmin-Barg condition. Based on Theorem 2.1, we obtain the following result.\\

\begin{proposition}\label{P-5.0}
    Let $k \ge 5$ and $u_1, u_2, \dots, u_t$ be integers satisfying $k-1 > u_t > \cdots > u_1 \geq 3$ and $\sum_{i=1}^{t}u_i <k$. The binary Solomon-Stiffler codes $[2^k - 1 - \sum_{i=1}^t (2^{u_i} - 1),\ k,\ 2^{k-1} - \sum_{i=1}^t 2^{u_i - 1}]_2$ code is a minimal self-orthogonal code satisfying the Ashikhmin-Barg condition. A family of binary minimal codes with parameters $$[2^k +2^{k-1} + t - 1 - \sum_{i=1}^t 2^{u_i+1},\ k,\ 2^{k-1} - \sum_{i=1}^t 2^{u_i - 1}]_2$$ are constructed. These self-orthogonal minimal binary codes violate the Ashikhmin-Barg condition.
\end{proposition}

{\bf Proof.} The simplex complementary codes of binary doubly even codes are doubly even, then the conclusion follows from Corollary 2.1.

\begin{example}{\rm When $k=5$, $t=1$, $u_1=3$, the corresponding Solomon-Stiffler code is a linear $[24,5,12]_2$ code with the weight enumerator $1+28z^{12}+3z^{16}$. Then a self-orthogonal binary minimal $[32,5,12]_2$ code with the weight enumerator $1+14z^{12}+z^{16}+14z^{20}+2z^{24}$ is constructed. It violates the Ashikhmin-Barg condition, since $\frac{w_{min}}{w_{max}}=\frac{1}{2}$. The optimal minimum weight of a $[32,5]_2$ code is $16$, see \cite{Grassl}.}
\end{example}

For $m\ge3$, the binary simplex $[2^m-1,m, 2^{m-1}]_2$ code is a minimal self-orthogonal code because it satisfies $\frac{2^m-1}{2^m-1}>\frac{1}{2}$ and its only weight $2^{m-1}$ is divisible by four.\\

\begin{proposition}\label{P-5.1}
    When $m\ge3$, a family of binary minimal self-orthogonal linear codes violating the Ashikhmin-Barg condition with parameters $[3\cdot2^{m-1}-1,m,2^{m-1}]_2$ and the maximum weight $2^m$ is constructed.
\end{proposition}

\begin{example}
    {\rm  When $m=3$, the simplex binary $[7,3,4]_2$ code is a self-orthogonal minimal code satisfying the Ashikhmin-Barg condition $\frac{4}{4}>\frac{1}{2}.$ We construct an explicit self-orthogonal minimal $[11,3,4]_2$ code with the maximum weight $8$. This binary self-orthogonal minimal code violates the Ashikhmin-Barg condition, since $\frac{4}{8}\leq \frac{1}{2}.$ The optimal minimum weight of a binary linear $[11,3]_2$ code is $6$.}
\end{example}

Let $m\ge 3$ be a positive integer. A family of binary three-weight linear codes with parameters $[2^{2m-2}+2^{m-1}-1, 2m-1, 2^{2m-3}]_2$ and nonzero weights $2^{2m-3}$, $2^{2m-3}+2^{m-2}$ and $2^{2m-3}+2^{m-1}$ was constructed in \cite[Theorem 3]{Hu2021}. When $m=3$, the code has parameters $[19,5,8]_2$ with the maximum weight $12$. From Theorem 3.1, we construct a linear $[23,5,8]_2$ code with the maximum weight $16$. This is a minimal code violating the Ashikhmin-Barg condition. The corresponding optimal minimum weight is $11$, see \cite{Grassl}.\\

For $m \ge 4$, all three nonzero weights are divisible by $4$, these binary codes are self-orthogonal. From Theorem 2.1, we have the following result.\\

\begin{proposition}\label{P-5.2}
    Let $m \ge 3$ be a positive integer. A family of $[3\cdot2^{2m-3}-1, 2m-1, 2^{2m-3}]_2$ codes with the maximum weight $2^{2m-2}$ is constructed. These codes are minimal codes violating the Ashikhmin-Barg condition. Moreover,  when $m \ge 4$, they are self-orthogonal.
\end{proposition}

\begin{example}
    {\rm When $m=4$, from the binary self-orthogonal $[71,7,32]_2$ code with weights $32$, $36$ and $40$, we construct a linear $[95,7,32]_2$ code with weights $32$, $36$, $40$, $56$, $60$ and $64$. This is a binary self-orthogonal minimal code violating the Ashikhamin-Barg condition. The corresponding optimal $[95,7,47]_2$ was documented in \cite{Grassl}. }
\end{example}

Let $m\geq3$ be a positive integer. A family of two-weight binary projective linear $[2^{2m-3}+2^{m-2}-1,2m-2,2^{2m-4}]_2$ codes with nonzero weights $2^{2m-4}$ and $2^{2m-4}+2^{m-2}$ was constructed in \cite[Theorem 5.2, Corollary 5.4]{WZZ}. When $m=3$, the code has parameters $[9,4,4]_2$ the maximum weight $6$. Then a minimal $[11,4,4]_2$ code violating the Ashikhmin-Barg condition is constructed.\\

When $m \ge 4$, both weights are divisible by $4$. Then these codes are self-orthogonal. From Theorem 2.1, we have the following result.\\

\begin{proposition}\label{P-5.4}
    Let $m\geq 3$ be a positive integer. A family of $[3 \cdot 2^{2m-4}-1,2m-2,2^{2m-4}]_2$ codes with the maximum weight $2^{2m-3}$ is constructed. These codes are minimal codes violating the Ashikhmin-Barg condition. Moreover, when $m \ge 4$, they are self-orthogonal.
\end{proposition}

\begin{example}
    {\rm When $m=4$, from the binary self-orthogonal $[35,6,16]_2$ code with weights $16$ and $20$, we construct a linear $[47,6,16]_2$ code with weights $16$, $20$, $28$ and $32$. This self-orthogonal minimal codes violates the Ashikhmin-Barg condition. The optimal minimum weight of the corresponding $[47,6]_2$ code is $23$, see \cite{Grassl}.}
\end{example}

In Table \ref{tab:5-1}, we list binary minimal self-orthogonal codes violating the Ashikhmin-Barg condition constructed in Proposition 5.1 - 5.4. Here, $\bC$ denotes the original self-orthogonal code satisfying the Ashikhmin-Barg condition, while $\bC'$ denotes the self-orthogonal minimal binary code violating the Ashikhmin-Barg condition.\\

\begin{longtable}{|l|l|l|l|l|}
\caption{\label{tab:5-1} Binary self-orthogonal minimal codes }\\ \hline
$\bC$  & $\bC'$      & Weights  & Optimal parameters & References \\ \hline
$[7,3,4]_2$ &  $[11,3,4]_2$   & $[4,8]$   &   $[11,3,6]_2$ & Proposition 5.2\\ \hline
$[16,4,8]_2$ &  $[23,4,8]_2$   & $[8,16]$   &   $[23,4,12]_2$ & Proposition 5.2\\ \hline
$[24,5,12]_2$ &  $[32,5,12]_2$   & $[12,16,20,24]$   &   $[32,5,16]_2$ & Proposition 5.1\\ \hline
$[31,5,16]_2$ & $[47,5,16]_2$ & $[16,32]$   &   $[47,5,24]_2$ & Proposition 5.2\\ \hline
$[35,6,16]_2$ & $[47,6,16]_2$   & $[16,20,28,32]$ & $[47,6,23]_2$ & Proposition 5.4\\ \hline
$[48,6,24]_2$ & $[64,6,24]_2$   & $[24,32,40,48]$   &   $[64,6,32]_2$ & Proposition 5.1\\ \hline
\end{longtable}

\section{Weight distributions of some minimal codes violating the Ashikhmin-Barg condition}

In this section, we give a general result about weight distributions of minimal codes violating the Ashikhmin-Barg condition constructed in this paper. Since weight distributions of simplex complementary codes were calculated in \cite{ChenXie}, then weight distributions of minimal codes violating the Ashikhmin-Barg condition constructed in this paper can be determined completely, if the weight distribution of the projective linear code is known.

\begin{theorem}\label{Thm-6.1}
    Let $\mathbf{C}' \subset \mathbf{F}_q^{n+n'}$ be the minimal linear code constructed in Theorem 2.1 from an $[n, k, d]$ linear code $\mathbf{C}$ with the generator matrix ${\bf G}$ and the weight distribution $[A_0,A_{w_1},A_{w_2},\dots,A_{w_t}]$. Let $\bC^* \subset \mathbb{F}_q^n$ denote the $(k-1)$-dimensional subcode of $\bC$ generated by removing the first row of ${\bf G}$, and assume its weight distribution is $[A^*_0,A^*_{w_1}, A^*_{w_2}, \dots, A^*_{w_t}]$. Then the set of nonzero weights in $\bC'$ is $$\{w_1,w_2,\dots,w_t \} \cup \{w_1+n',w_2+n',\dots,w_t+n' \}.$$ More precisely, the weight distribution of $\bC'$ is given by:
    $$[A'_0,A'_{w_1}, \dots, A'_{w_t}, A'_{w_1 + n'}, \dots, A'_{w_t + n'}],$$
    where
    $$A'_0=1, \ A'_{w_i} = A^*_{w_i} \ \text{and} \ A'_{w_i + n'} = A_{w_i} - A^*_{w_i}, \quad 1 \le i \le t.$$
\end{theorem}

\begin{proof}
{\rm
    Let ${\bf G}$ be a generator matrix of the linear code ${\bf C}$, with the first row ${\bf r}_1$ equal to a codeword $(c_1,\ldots, c_n)$ with the maximum weight $w_{max}$, and the second row ${\bf r}_2$ equal to a codeword with the minimum weight $w_{min}$. Let ${\bf G}'$ be a $k \times (n+n')$ matrix with the first row ${\bf r}_1'$ equal to $(c_1',\ldots, c_n',c_1, \ldots, c_n)$, where $c_1', \ldots, c_n'$ are nonzero elements in ${\bf F}_q$, and the other rows are of the form ${\bf r}_i'=({\bf 0}, {\bf r}_i)$, where ${\bf 0}$ is the zero vector in ${\bf F}_q^{n'}$, ${\bf r}_i$ is $i$-th row of ${\bf G}$, for $i=2,\ldots,k$. Let $\bC^*$ be the linear subcode of ${\bf C}'$ span by  ${\bf r}_2', \dots,{\bf r}_k' $.

    We analyze the weight distribution of $\bC'$ by considering the two cases of codewords ${\bf x} \in \bC'$.\\

    1) ${\bf x}$ is the linear combination of rows ${\bf r}_2',\ldots, {\bf r}_k'$ of ${\bf G}'$. In this case, ${\bf x} \in \bC^*$. The first $n'$ coordinates of ${\bf x}$ are zero, and the last $n$ coordinates form a codeword in $\bC^*$. Hence, the weights of such codewords belongs to $\{w_1, \dots, w_t\}$. Therefore, $$A'_{w_i} = A^*_{w_i}, \quad 1 \le i \le t.$$

    2) ${\bf x}$ is the linear combination of rows ${\bf r}_1', \ldots, {\bf r}_k'$ of ${\bf G}'$. Suppose that $${\bf x}=\mu_1{\bf r}_1'+\sum_{i=2}^k \mu_i {\bf r}_i',$$ where $\mu_1\neq 0$. Let ${\bf r}_1' = ({\bf a},{\bf r}_1)$ and ${\bf r}_i' = ({\bf 0}, {\bf r}_i)$ for $2 \le i \le k$ and ${\bf a} \in ({\bf F}_q^*)^{n'}$, then

    $${\bf x} = \left( \mu_1 {\bf a}, \mu_1 {\bf r}_1 + \sum_{i=2}^k \mu_i {\bf r}_i \right) = ({\bf x}_{1}, {\bf x}_{2}).$$

    Since $\mu_1 \ne 0$ and ${\bf a} \ne {\bf 0}$, ${wt}({\bf x}_{1}) = n'$. Furthermore, ${\bf x}_{2} \in \bC$, then ${\bf x}_2$ has weight $w_i \in \{w_1, \dots, w_t\}$. It follows that
    $$wt({\bf x}) =  wt({\bf x}_{1}) + wt({\bf x}_{2})= n' + w_i,\ 1 \le i \le t.$$

    There are $A_{w_i}$ codewords of weight $w_i$ in $\bC$, and $A^*_{w_i}$ of them appear in Case 1, then the number of codewords in $\bf C'$ of weight $w_i + n'$ is $$A'_{w_i + n'} = A_{w_i} - A^*_{w_i}, \quad 1 \le i \le t.$$

    The conclusion is proved.
}
\end{proof}

The weight distribution of the $q$-ary simplex $[\frac{q^m-1}{q-1},m, q^{m-1}]_q$ code, is $A_0=1$ and $A_{q^{m-1}}=q^{m-1}$. It is obvious that the $q$-ary simplex code is a minimal code satisfying the Ashikhmin-Barg condition. Applying Theorem 2.1 with the extension length $n'=\left \lceil \frac{q^{m-1}}{q-1}  \right \rceil$, we construct a minimal $[\frac{q^{m}-1}{q-1}+n', m, q^{m-1} ]_q$ code violating the Ashikhmin-Barg condition. Its nonzero weights are $q^{m-1}$ and $q^{m-1}+n'$. For example, when $q=2$, the two-weight linear $[5,2,2]_2$, $[11,3,4]_2$ and $[23,4,8]_2$ codes are obtained. The corresponding optimal minimum weights are $3$, $6$ and $12$. Similarly, two-weight linear codes $[6,2,3]_3$, $[18,3,9]_3$, $[7,2,4]_4$, $[27,3,16]_4$, $[8,2,5]_5$, $[10,2,7]_7$, $[11,2,8]_8$ and $[12,2,9]_9$ are constructed. The corresponding optimal minimum weights are $4$, $12$, $5$, $20$, $6$, $8$, $9$ and $10$, see \cite{Grassl}.

\begin{proposition}\label{P-6.1}
    Let $q$ be a prime power, and $m \ge 2$ be an integer. Set $n'=\left \lceil \frac{q^{m-1}}{q-1}  \right \rceil$. A family of minimal linear $[\frac{q^{m}-1}{q-1}+n', m, q^{m-1} ]_q$ codes with the weight distribution given in Table \ref{tab:6-1} is constructed from the $q$-ary simplex $[\frac{q^m-1}{q-1},m, q^{m-1}]_q$ code. These codes violate the Ashikhmin-Barg condition.
\end{proposition}

As a corollary of Proposition \ref{P-6.1}, we have the following result.

\begin{corollary}\label{C-6.1}
    Let $q=3$, $m\ge2$ be an integer. Set $n'=\left \lceil \frac{3^{m-1}}{2}  \right \rceil+1$. A family of minimal ternary self-orthogonal linear $[\frac{3^{m}-1}{2}+n', m, 3^{m-1} ]_3$ codes with the weight distributions given in Table \ref{tab:6-1} is constructed from the ternary simplex $[\frac{3^m-1}{2},m, 3^{m-1}]_q$ code. These codes violate the Ashikhmin-Barg condition.
\end{corollary}

\begin{example}
    {\rm  When $m=2$, a two-weight minimal ternary code with parameters $[7,2,3]_3$ and the maximum weight $6$ is constructed from the simplex $[4,2,3]_3$ code. This code is a self-orthogonal minimal code violating the Ashikhmin-Barg condition. The corresponding optimal weight is $5$, see \cite{Grassl}.}
\end{example}

\begin{longtable}{|l|l|}
\caption{\label{tab:6-1} Weight distribution of the codes in Proposition 6.1}\\ \hline
Weight          & Weight distribution \\ \hline
$0$             & $1$ \\ \hline
$q^{m-1}$       & $q^{m-1}-1$ \\ \hline
$q^{m-1}+n'$    & $(q-1) \cdot q^{m-1}$ \\ \hline
\end{longtable}

Let $m \geq 3$ be an odd integer. The dual codes of the primitive BCH codes with parameters $[2^m-1, 2^m-1-2m, 5]_2$ is a family of three-weight linear codes with parameters $[2^m-1, 2m, 2^{m-1}-2^{\frac{m-1}{2}}]_2$. The weight distributions are
    $$A_{2^{m-1}-2^{\frac{m-1}{2}}}=(2^m-1)(2^{m-2}+2^{\frac{m-3}{2}}),$$
    $$A_{2^{m-1}}=(2^m-1)(2^{m-1}+1),$$
    and
    $$A_{2^{m-1}+2^{\frac{m-1}{2}}}=(2^m-1)(2^{m-2}-2^{\frac{m-3}{2}}),$$
    see \cite[Theorem 4]{carlet}.
For all odd $m \geq 5$, these codes satisfy the Ashikhmin-Barg criterion for the minimality. From Theorem 2.1, these codes can be transformed to minimal codes violating the Ashikhmin-Barg condition.\\

\begin{proposition}\label{P-6.2}
    Let $m \geq 5$ be an odd integer. Set $n'=2^{m-1}-3 \cdot 2^{\frac{m-1}{2}}$. A family of minimal linear $[2^m + n' -1 , 2m, 2^{m-1}-2^{\frac{m-1}{2}}]_2$ codes with the weight distributions given in Table \ref{tab:6-2} is constructed. These codes violate the Ashikhmin-Barg condition.
\end{proposition}

\begin{longtable}{|l|l|}
\caption{\label{tab:6-2} Weight distribution of the codes in Proposition 6.2}\\ \hline
Weight                              & Weight distribution \\ \hline
$0$                                 & $1$ \\ \hline
$2^{m-1}-2^{\frac{m-1}{2}}$         & $(2^{m-1}+2^{\frac{m-1}{2}}-1)(2^{m-2}+2^{\frac{m-3}{2}})$ \\ \hline
$2^{m-1}$                           & $(2^{m-1}-1)(2^{m-1}+1)$ \\ \hline
$2^{m-1}+2^{\frac{m-1}{2}}$         & $(2^{m-1}-2^{\frac{m-1}{2}}-1)(2^{m-2}-2^{\frac{m-3}{2}})$ \\ \hline

$2^{m}-2^{\frac{m+3}{2}}$           & $(2^{m-1}-2^{\frac{m-1}{2}})(2^{m-2}+2^{\frac{m-3}{2}})$ \\ \hline
$2^{m}-3 \cdot 2^{\frac{m-1}{2}}$   & $(2^{m-1})(2^{m-1}+1)$ \\ \hline
$2^{m}-2^{\frac{m+1}{2}}$           & $(2^{m-1}+2^{\frac{m-1}{2}})(2^{m-2}-2^{\frac{m-3}{2}})$ \\ \hline

\end{longtable}

\begin{example}
    {\rm  Let $m=5$, the dual code of the primitive BCH $[31, 21, 5]_2$ code has parameters $[31,10,12]_2$ and the weight enumerator $1+310z^{12}+527z^{16}+186z^{20}$. With $n'=2\times12-20=4$, a four-weight $[35,10,12]_2$ code with the weight distribution $A'_{12}=(31-12)\times10=190$, $A'_{16}=12\times10 + (31-16)\times17=375$, $A'_{20}=16\times17 + (31-20)\times6 = 338$ and $A'_{24}=20\times6=120$ is constructed. This was confirmed by Magma. Notably, the minimal linear $[35, 10, 12]_2$ code is optimal, see \cite{Grassl}.}
\end{example}

\section{Conclusions}

Since the publication of \cite{CH,DHZ}, there have been many papers constructing several infinite families minimal linear codes violating the Ashikhmin-Barg condition, from few-weight linear codes, Boolean functions, posets, or partial difference sets. In this paper, we proved that minimal linear codes violating the Ashikhmin-Barg condition can be constructed from arbitrary projective linear codes. Then infinitely many families of minimal linear codes violating the Ashikhmin-Barg condition were constructed. Moreover, infinitely many families of self-orthogonal binary minimal linear codes violating the Ashikhmin-Barg condition were also given. Many minimal codes violating the Ashikhmin-Barg condition constructed in this paper have their parameters close to optimal or best known ones. Weight distributions of minimal linear codes violating the Ashikhmin-Barg condition constructed in this paper can be determined explicitly.\\

\begin{appendices}
\section{Minimal codes violating Ashikhmin-Barg condition constructed in this paper}\label{A-1}

In the following tables, ${\bf C}'$ is the minimal linear code violating the Ashikhmin-Barg condition. These minimal codes are compared with optimal or best known codes in \cite{Grassl}.

    \begin{longtable}{|l|l|l|l|l|}
    \caption{\label{tab:app} Minimal codes violating the Ashikhmin-Barg condition constructed in this paper }\\ \hline
    $q$ & $\bC$ & $\bC'$ & Optimal or best known parameters  & References \\ \hline
    $2$ & $[3,2,2]_2$ & $[5,2,2]_2$  & $[5,2,3]_2$  & Proposition \ref{P-6.1}\\ \hline
    $2$ & $[7,3,4]_2$ & $[11,3,4]_2$  & $[11,3,6]_2$  & Proposition \ref{P-6.1}\\ \hline
    $2$ & $[9,4,4]_2$ & $[11,4,4]_2$  & $[11,4,5]_2$  & Proposition \ref{P-5.4}\\ \hline
    $2$ & $[11,4,5]_2$ & $[13,4,5]_2$  & $[13,4,6]_2$  & Proposition \ref{P-4.9}\\ \hline
    $2$ & $[15,5,6]_2$ & $[17,5,6]_2$  & $[17,5,8]_2$  & Proposition \ref{P-4.4}\\ \hline
    $2$ & $[15,6,6]_2$ & $[17,6,6]_2$  & $[17,6,7]_2$  & Proposition \ref{P-4.2}\\ \hline
    $2$ & $[16,5,6]_2$ & $[18,5,6]_2$  & $[18,5,8]_2$  & Proposition \ref{P-4.8}\\ \hline
    $2$ & $[15,4,8]_2$ & $[23,4,8]_2$  & $[23,4,12]_2$  &  Proposition \ref{P-6.1}\\ \hline
    $2$ & $[19,5,8]_2$ & $[23,5,8]_2$  & $[23,5,11]_2$ & Proposition \ref{P-5.2}\\ \hline
    $2$ & $[22,5,10]_2$ & $[26,5,10]_2$  & $[26,5,12]_2$  & Proposition \ref{P-4.7}\\ \hline
    $2$ & $[23,5,10]_2$ & $[27,5,10]_2$  & $[27,5,13]_2$  & Proposition \ref{P-4.5}\\ \hline
    $2$ & $[20,5,8]_2$ & $[24,5,8]_2$  & $[24,5,12]_2$  & Proposition \ref{P-4.8}\\ \hline
    $2$ & $[25,5,12]_2$ & $[32,5,12]_2$  & $[32,5,16]_2$  & Proposition \ref{P-4.9}\\ \hline
    $2$ & $[31,6,12]_2$ & $[35,6,12]_2$  & $[35,6,16]_2$  & Proposition \ref{P-4.4}\\ \hline
    $2$ & $[31,10,12]_2$ & $[35,10,12]_2$ & $[35,10,12]_2$  &  Proposition \ref{P-6.2}\\ \hline 
    $2$ & $[27,5,13]_2$ & $[37, 5, 13]_2$  & $[37, 5, 18]_2$  & Proposition \ref{P-4.9}\\ \hline

    $3$ & $[4,2,3]_3$ & $[6,2,3]_3$ & $[6,2,4]_3$  & Proposition \ref{P-6.1}\\ \hline
    $3$ & $[4,2,3]_3$ & $[7,2,3]_3$ & $[7,2,5]_3$  & Corollary \ref{C-6.1}\\ \hline 
    $3$ & $[9,3,5]_3$ & $[10,3,5]_3$ & $[10,3,6]_3$  & Proposition \ref{P-4.10}\\ \hline
    $3$ & $[8,2,6]_3$ & $[12,2,6]_3$ & $[12,2,8]_3$  & Theorem 2.1\\ \hline 
    $3$ & $[13,3,9]_3$ & $[18,3,9]_3$ & $[18,3,12]_3$  & Proposition \ref{P-6.1}\\ \hline 
    $3$ & $[13,3,9]_3$ & $[19,3,9]_3$ & $[19,3,12]_3$  & Corollary \ref{C-6.1}\\ \hline 
    $3$ & $[26,6,15]_3$ & $[29,6,15]_3$ & $[29,6,16]_3$  & Theorem 2.1\\ \hline  
    $3$ & $[40,4,24]_3$ & $[47,4,24]_3$ & $[47,4,30]_3$  & Theorem 2.1\\ \hline  
    $3$ & $[52,6,30]_3$ & $[56,6,30]_3$ & $[56,6,36]_3$  &  Theorem 2.1\\ \hline  
    $3$ & $[56,6,30]_3$ & $[60,6,30]_3$ & $[60,6,36]_3$  & Theorem 2.1\\ \hline  
    $3$ & $[104,8,60]_3$ & $[111,8,60]_3$ & $[111,8,66]_3$  & Theorem 2.1\\ \hline  

    $4$ & $[5,2,4]_4$ & $[7,2,4]_4$ & $[7,2,5]_4$  & Proposition \ref{P-6.1}\\ \hline
    $4$ & $[21,3,16]_4$ & $[27,3,16]_4$ & $[27,3,20]_4$  & Proposition \ref{P-6.1}\\ \hline
    $5$ & $[6,2,5]_5$	& $[8,2,5]_5$ & $[8,2,6]_5$  & Proposition \ref{P-6.1}\\ \hline
    $5$ & $[23,3,19]_5$ & $[28,3,19]_5$  &  $[28,3,22]_5$ & Proposition \ref{P-4.10}\\ \hline

    $7$ & $[8,2,7]_7$	 & $[10,2,7]_7$ & $[10,2,8]_7$  & Proposition \ref{P-6.1}\\ \hline

    $8$ & $[9,2,8]_8$	 & $[11,2,8]_8$ & $[11,2,9]_8$  & Proposition \ref{P-6.1}\\ \hline

    $9$ & $[10,2,9]_9$ & $[12,2,9]_9$ & $[12,2,10]_9$  & Proposition \ref{P-6.1}\\ \hline

    \end{longtable}	

\begin{longtable}{|l|l|l|l|}
    \caption{\label{tab:from-codetable} $q$-ary minimal codes violating the Ashikhmin-Barg condition constructed from best known codes in \cite{Grassl}}\\ \hline
    $q$ & $\bf C$   & $\bf C'$  & Optimal or best known parameters \\ \hline
    $2$ & $[32, 7, 14]_2$ & $[38, 7, 14]_2$ & $[38, 7, 16]_2$ \\ \hline 
    $2$ & $[35, 8, 15]_2$ & $[41, 8, 15]_2$ & $[41, 8, 17]_2$ \\ \hline 
    $2$ & $[38, 9, 16]_2$ & $[42, 9, 16]_2$ & $[42, 9, 17]_2$ \\ \hline 
    $2$ & $[43, 9, 18]_2$ & $[47, 9, 18]_2$ & $[47, 9, 20]_2$ \\ \hline 
    $2$ & $[45, 9, 19]_2$ & $[51, 9, 19]_2$ & $[51, 9, 22]_2$ \\ \hline 
    $2$ & $[62,9,28]_2$ & $[74,9,28]_2$ & $[74,9,32]_2$ \\ \hline
    $2$ & $[112,15,45]_2$ & $[125,15,45]_2$ & $[125,15,53]_2$ \\ \hline

    $3$ & $[26,6,15]_3$ & $[28,6,15]_3$ & $[28,6,15]_3$ \\ \hline
    $3$ & $[66,8,38]_3$ & $[69,8,38]_3$ & $[69,8,39]_3$ \\ \hline
    $3$ & $[110,9,64]_3$ & $[112,9,64]_3$ & $[112,9,66]_3$ \\ \hline
    $3$ & $[182,15,99]_3$ & $[184,15,99]_3$ & $[184,15,99]_3$ \\ \hline  
    $3$ & $[208,14,117]_3$ & $[216,14,117]_3$ & $[216,14,120]_3$ \\ \hline  

    $4$ & $[10,2,8]_4$ & $[13,2,8]_4$ & $[13,2,10]_4$ \\ \hline
    $4$ & $[26,3,19]_4$ & $[28,3,19]_4$ & $[28,3,20]_4$ \\ \hline
    $4$ & $[57,6,38]_4$ & $[58,6,38]_4$ & $[58,6,39]_4$ \\ \hline
    $4$ & $[73,7,49]_4$ & $[75,7,49]_4$ & $[75,7,51]_4$ \\ \hline

    $5$ & $[27,3,21]_5$ & $[29,3,21]_5$ & $[29,3,23]_5$ \\ \hline 
    $5$ & $[33,3,25]_5$ & $[34,3,25]_5$ & $[34,3,26]_5$ \\ \hline 
    $5$ & $[38,3,30]_5$ & $[41,3,30]_5$ & $[41,3,32]_5$ \\ \hline 
    $5$ & $[57,4,44]_5$ & $[58,4,44]_5$ & $[58,4,45]_5$ \\ \hline 
    $5$ & $[62,6,45]_5$ & $[64,6,45]_5$ & $[64,6,45]_5$ \\ \hline   
    $5$ & $[66,4,51]_5$ & $[67,4,51]_5$ & $[67,4,52]_5$ \\ \hline   
    $5$ & $[70,4,55]_5$ & $[74,4,55]_5$ & $[74,4,58]_5$ \\ \hline   
    $5$ & $[83,5,63]_5$ & $[86,5,63]_5$ & $[86,5,65]_5$ \\ \hline   
    $5$ & $[90,5,69]_5$ & $[93,5,69]_5$ & $[93,5,70]_5$ \\ \hline   
    $5$ & $[98,5,75]_5$ & $[102,5,75]_5$ & $[102,5,78]_5$ \\ \hline   
    $5$ & $[115,5,88]_5$ & $[117,5,88]_5$ & $[117,5,90]_5$ \\ \hline   
    $5$ & $[124,9,90]_5$ & $[127,9,90]_5$ & $[127,9,91]_5$ \\ \hline   

    $7$ & $[16,2,14]_7$ & $[19,2,14]_7$ & $[19,2,16]_7$ \\ \hline   
    $7$ & $[19,2,16]_7$ & $[20,2,16]_7$ & $[20,2,17]_7$ \\ \hline   
    $7$ & $[21,2,18]_7$ & $[23,2,18]_7$ & $[23,2,20]_7$ \\ \hline   
    $7$ & $[33,2,28]_7$ & $[36,2,28]_7$ & $[36,2,31]_7$ \\ \hline   
    $7$ & $[42,2,36]_7$ & $[46,2,36]_7$ & $[46,2,40]_7$ \\ \hline   
    $7$ & $[51,3,43]_7$ & $[53,3,43]_7$ & $[53,3,45]_7$ \\ \hline   
    $7$ & $[67,3,56]_7$ & $[68,3,56]_7$ & $[68,3,57]_7$ \\ \hline   
    $7$ & $[71,3,60]_7$ & $[72,3,60]_7$ & $[72,3,61]_7$ \\ \hline   
    $7$ & $[76,3,64]_7$ & $[77,3,64]_7$ & $[77,3,65]_7$ \\ \hline   
    $7$ & $[80,3,68]_7$ & $[83,3,68]_7$ & $[83,3,70]_7$ \\ \hline   
    $7$ & $[90,3,77]_7$ & $[96,3,77]_7$ & $[96,3,82]_7$ \\ \hline   

    $8$ & $[23,2,20]_8$ & $[24,2,20]_8$ & $[24,2,21]_8$ \\ \hline   
    $8$ & $[29,2,25]_8$ & $[31,2,25]_8$ & $[31,2,27]_8$ \\ \hline   
    $8$ & $[103,4,88]_8$ & $[104,4,88]_8$ & $[104,4,88]_8$ \\ \hline   

    $9$ & $[28,2,25]_9$ & $[31,2,25]_9$ & $[31,2,27]_9$ \\ \hline 
    $9$ & $[41,2,36]_9$ & $[43,2,36]_9$ & $[43,2,38]_9$ \\ \hline 
    $9$ & $[56,3,48]_9$ & $[57,3,48]_9$ & $[57,3,49]_9$ \\ \hline 
    $9$ & $[62,3,54]_9$ & $[64,3,54]_9$ & $[64,3,56]_9$ \\ \hline 
\end{longtable}

\end{appendices}


\begin{thebibliography}{100}

\bibitem{Alfarano} G. N. Alfarano, M. Borello, A. Neri and A. Ravaganani, Three combinatorial aspects on minimal codes, SIAM Journal on Discrete Mathematics, vol. 36, pp. 461-489, 2022.


\bibitem{Alon} N. Alon, A. Bishnoi, S. Das and A. Neri, Strong blocking sets and minimal codes from expander graphs, arXiv:2305.15297, 2023.

\bibitem{AB} A. Ashikhmin and A. Barg, Minimal vectors in linear codes, IEEE Transactions on Information Theory, vol. 44, no. 5, pp. 2010-2017, 1998.

\bibitem{BB1} D. Bartoli and M. Bonini, Minimal linear codes in odd characteristic,
IEEE Transactions on Information Theory, vol. 65, no. 7, pp. 4152-4155, 2019.





 \bibitem{BB}   M. Bonini and M. Borello, Minimal linear codes arising from
blocking sets, Journal of Algebraic Combinatorics, vol. 53, pp. 327-341, 2021.

\bibitem{carlet} C. Carlet, P. Charpin and V. Zinoviev, Codes,bent functions and permutations suitable for DES-like cryptosystems, Designs, Codes and Cryptography, vol. 15, pp. 125-156, 1998.

\bibitem{CH} S. Chang and J. Y. Hyun, Linear codes from simplicial complexes, Designs, Codes
and Cryptography, vol. 86, pp. 2167-2181, 2018.

\bibitem{CDMT} H. Chen, C. Ding, S. Mesanager and C. Tang, A noval application of Boolean functions with high algebraic immunity in minimal codes, IEEE Transactions on Information Theory, vol. 67, no. 10, pp. 6856-6867, 2021.

\bibitem{ChenXie} H. Chen and C. Xie, Projective linear codes and their simplex complementary codes, Journal of Algebra, vol. 673, pp. 304-320, 2025.

\bibitem{ChenXiePan} Y. Chen, C. Xie and X. Pan, New families of optimal, almost optimal and near optimal few-weight linear codes, submitted, 2025.


\bibitem{CL} G. D. Cohen and A. Lempel, Linear intersecting codes, Discrete Mathematics, vol. 56, pp. 35-43, 1984.

\bibitem{Ding} C. Ding, Linear codes from some $2$-designs, IEEE Transactions on Information Theory, vol. 61, no. 6, pp. 3265-3275, 2015.

\bibitem{DD} K. Ding, and C. Ding, A class of two-weight and three-weight codes and their applications in secret sharing, IEEE Transactions on Information Theory, vol. 61, no. 11, pp. 5835-5842, 2015.

\bibitem{DD1}  K. Ding and C. Ding, Binary linear codes with three weights, IEEE Communication Letters, vol. 18, no. 11, pp. 1679-1882, 2014.

\bibitem{DHZ} C. Ding, Z. Heng and Z. Zhou, Minimal binary linear codes, IEEE Transactions on Information Theory, vol. 64, pp. 6536-6545, 2018.


\bibitem{Grassl} M. Grassl, Bounds on the minimum distance of linear codes and quantum codes,
Online available at http://www.codetables.de.
		
\bibitem{Griesmer} J. H. Griesmer, A bound for error-correcting code, IBM Journal of Research and Development., vol. 4, pp. 532-542, 1960.





\bibitem{HengYue} Z. Heng and Q. Yue, Several classes of cylic codes with either optimal three weights or a few weights, IEEE Transactions on Information Theory, vol. 62, no. 8, pp. 4501-4513, 2016.


\bibitem{HDZ}  Z. Heng, C. Ding, and Z. Zhou, Minimal linear codes over finite fields,
Finite Fields and Theie Applications, vol. 54, pp.176-916, 2018.


\bibitem{Hu2021} Z. Hu, L. Wang, N. Li and X. Zeng, Several classes of linear codes with few weights from the closed butterfly structure. Finite Fields and Their Applications, vol. 76, no. 2, pp. 101926, 2021.

\bibitem{Hu} Z. Hu, Y. Xu, N. Li, X. Zeng, L. Wang and X. Tang, New constructions of optimal linear codes from simplicial complexes, IEEE Transactions on Theory, vol. 70, no. 3, pp. 1823-1835, 2024.

	




\bibitem{Huang} T.-Y. Huang, Decoding linear block codes for minimizing word error rate, IEEE Transactions on Information Theory, vol. 25, pp. 733-737, 1979.

\bibitem{HP} W. C. Huffman and V. Pless, Fundamentals of error-correcting codes, Cambridge University Press, Cambridge, U. K., 2003.

\bibitem{Hyun} J. Y. Hyun, J. Lee and Y. Lee, Infinite families of optimal linear codes constructed from simplicial complexes,  IEEE Transactions on Information Theory, vol. 66, no. 11, pp. 6762-6773, 2020.

\bibitem{Hyun1} J. Y. Hyun, H. K. Kim, Y. Wu and Q. Yue, Optimal minimal linear codes from posets, Designs, Codes and Cryptography, vol. 88, pp. 2475-2492, 2020.




\bibitem{Kan} Y. Li, H. Kan and L. Zhang, Minimal binary linear codes from vectorial Boolean functions, IEEE Transactions on Information Theory, vol. 69, pp. 2955-2968, 2023.

\bibitem{Kasami} T. Kasami, Weight distribution formula for some class of cyclic codes, Technical Report R-285, (AD 632574), Coordinated Science Laboratory, University of Illinois, Urbana, IL, 1966

\bibitem{LiYue}  X. Li and Q. Yue, Four classes of minimal binary linear codes
with $\frac{w_{min}}{w_{max}} <1/2$ derived from Boolean functions, Designs, Codes and Cryptography, vol. 88, pp. 257-271, 2019.

\bibitem{Lint} J. H. van Lint, Introduction to the coding theory, GTM 86, Third and Expanded Edition, Springer, Berlin, 1999.



\bibitem{MScode} F. J.  MacWilliams and N. J. A. Sloane, The Theory of error-correcting codes, 3rd Edition, North-Holland Mathematical Library, vol. 16. North-Holland, Amsterdam, 1977.



\bibitem{Sihem2} S. Mesnager and A. Sinak, Several classes of minimal linear codes with few weights from weakly regular plateaued functions, IEEE Transactions on Information Theory, vol. 66, pp. 2296-2310, 2020.



\bibitem{Sihem} S. Mesnager, Y. Qi, H. Ru, and C. Tang,  Minimal linear codes from characteristic functions, IEEE Transactions on Information Theory, vol. 66, pp. 5404-5413, 2020.





\bibitem{Sihem1} S. Mesanager, L. Qian X. Cao and M. Yuan, Sveral families of minimal binary linear codes from two-to-one functions, IEEE Transactions on Information Theory, vol. 69, pp. 3285-3301, 2023.


\bibitem{PRZW} E. Pasalic, R. Rodriguez, F. Zhang and Y. Wei, Several classes of minimal binary codes violating the Ashikhmin-Barg bound, Cryptography and Communications, vol. 13, pp. 637-659, 2021.


\bibitem{PRZW1} R. Rodriguez-Aldama, E. Pasalic, F. Zhang and Y. Wei, Minimal $p$-ary codes via direct sum of functions, non-covering permutations and subspaces of derivatives, IEEE Transactions on Information Theory, vol. 70, pp. 4445-4462, 2024.

\bibitem{SH}  M. Shi, S. Li, T. Helleseth, and J.-L. Kim, Binary self-orthogonal codes which meet the Griesmer bound or have optimal minimum distances, Journal of Combinatorial Theory, Series A, vol. 214, no. 101892, 2025.

\bibitem{Solomon} G. Solomon and J. J. Stiffler, Algebraically punctured cyclic codes, Information and Controll, vol. 8, pp. 170-179, 1965.







\bibitem{Tang}    C. Tang, Y. Qiu, Q. Liao, and Z. Zhou, Full characterization of minimal
linear codes as cutting blocking sets, IEEE Transactions on Information Theory, vol. 67, pp. 3690-3700, 2021.


\bibitem{TangLi} D. Tang and X. Li, A note on the minimal binary linear code, Cryptography and Communications, vol. 12, pp. 375-388, 2020.

\bibitem{TFL} R. Tao, T. Feng, and W. Li, A construction of minimal linear codes
from partial difference sets, IEEE Transactions on Information Theory, vol. 67, no. 6, pp. 3724-3734, 2021.

\bibitem{WZZ} X. Wang, D. Zheng and Y. Zhang, Binary linear codes with few weights from Boolean functions, Designs, Codes and Cryptography, vol. 89, pp. 2009-2030, 2021.

\bibitem{Xu} G. Xu and L. Qu, Three classes of minimal linear codes over finite fields of the odd characteristic, IEEE Transactions on Information Theory, vol. 65, no. 11, pp. 7067-7078, 2019.

\bibitem{XQC} G. Xu, L. Qu, and X. Cao, Minimal linear codes from Maiorana–McFarland functions, Finite Fields and their Application, vol. 65, no. 101688, 2020.

\bibitem{Zhang}  F. Zhang, E. Pasalic, R. Rodríguez, and Y. Wei,  Minimal binary linear
codes: A general framework based on bent concatenation, Designs, Codes
and Cryptography, vol. 90, pp. 1289-1318, 2022.


\bibitem{Yan} W. Zhang, H. Yan and H. Wei, Four families of minimal binary codes with $\frac{w_{min}}{w_{max}}<\frac{1}{2}$, Applied Algebra in Engineering, Communication and Computing, vol. 30, pp. 175-184, 2019.
\end{thebibliography}
\end{document}